\documentclass[a4paper,UKenglish,cleveref, autoref, thm-restate]{lipics-v2021}
\usepackage{booktabs}
\usepackage{xspace}
\usepackage{thmtools}
\usepackage{float}

\title{Computing Largest Subsets of Points Whose Convex Hulls have Bounded Area and Diameter}

\graphicspath{ {figures/} }

\author{Gianmarco Picarella}{Department of Information and Computing Sciences, Utrecht University, The Netherlands}{g.picarella@students.uu.nl}{https://orcid.org/0009-0006-3748-3343}{}
\author{Marc van Kreveld}{Department of Information and Computing Sciences, Utrecht University, The Netherlands}{M.J.vanKreveld@uu.nl}{https://orcid.org/0000-0001-8208-3468}{}
\author{Frank Staals}{Department of Information and Computing Sciences, Utrecht University, The Netherlands}{f.staals@uu.nl}{https://orcid.org/0009-0004-8522-1351}{}
\author{Sjoerd de Vries}{Department of Information and Computing Sciences, Utrecht University, The Netherlands \and Department of Digital Health, University Medical Center Utrecht, The Netherlands}{s.devries1@uu.nl}{https://orcid.org/0000-0002-5306-6797}{}

\newcommand{\opt}{\mbox{\scriptsize{\rm {opt}}}}

\titlerunning{Computing Largest Subsets of Points Whose Convex Hulls are Bounded}

\authorrunning{G. Picarella, M. van Kreveld, F. Staals, and S. de Vries} %TODO mandatory. First: Use abbreviated first/middle names. Second (only in severe cases): Use first author plus 'et al.'

\Copyright{Gianmarco Picarella, Marc van Kreveld, Frank Staals and Sjoerd de Vries.} %TODO mandatory, please use full first names. LIPIcs license is "CC-BY";  http://creativecommons.org/licenses/by/3.0/

 %Please choose ACM 2012 classifications from https://dl.acm.org/ccs/ccs_flat.cfm

  %TODO mandatory; please add comma-separated list of keywords

\relatedversion{} %optional, e.g. full version hosted on arXiv, HAL, or other respository/website
\supplement{\url{https://github.com/gianmarcopicarella/area-diameter-algorithm}}%optional, e.g. related research data, source code, ... hosted on a repository like zenodo, figshare, GitHub, ...
\nolinenumbers %uncomment to disable line numbering

\hideLIPIcs  %uncomment to remove references to LIPIcs series (logo, DOI, ...), e.g. when preparing a pre-final version to be uploaded to arXiv or another public repository

\EventEditors{John Q. Open and Joan R. Access}
\EventNoEds{2}
\EventLongTitle{42nd Conference on Very Important Topics (CVIT 2016)}
\EventShortTitle{CVIT 2016}
\EventAcronym{CVIT}
\EventYear{2016}
\EventDate{December 24--27, 2016}
\EventLocation{Little Whinging, United Kingdom}
\EventLogo{}
\SeriesVolume{42}
\ArticleNo{23}
\begin{document}

\maketitle
\begin{abstract}
  We study the problem of computing a convex region with bounded area and diameter that contains the maximum number of points from a given point set $P$. We show that this problem can be solved in $O(n^6k)$ time and $O(n^3k)$ space, where $n$ is the size of $P$ and $k$ is the maximum number of points in the found region.
  We experimentally compare this new algorithm with an existing algorithm that does the same but without the diameter constraint, which runs in $O(n^3k)$ time. For the new algorithm, we use different diameters.
  We use both synthetic data and data from an application in cancer detection, which motivated our research.
\end{abstract}

\section{Introduction}
\label{sec:Introduction}

Medical image data is examined by specialists to detect anomalies. For example, human tissue may contain mitotic cells,
%\frank{this may be a bit too technical for ESA?}
where dense areas are indicative of the presence of breast cancer. Therefore, pathologists used to count mitotic cells in small areas to establish a risk and treatment regime. To automate this process, mitotic cells are nowadays identified from images and converted to (planar) point sets. Dense areas are regions of a certain size and shape that contain many points, which need to be found and inspected.

According to medical specialists, a region of interest is convex, has an area of $2$--$4$ mm$^2$, a bounded aspect ratio, and has the maximum number of points~\cite{umcmodel}. More generally, the identification of well-shaped, small, dense areas in point sets is a common data analysis task.

In this paper we address the problem of computing---for a point set $P$---a convex region with the largest number of points from $P$, given an upper bound on the area and on the diameter. We present both algorithmic and experimental results. Earlier research showed that a convex region with at least $k$ points and the minimum area can be computed using dynamic programming in $O(n^3k)$ time and $O(n^2k)$ space. This algorithm can be used to compute the maximum number of points in a convex region of at most a given area, but the diameter constraint is not taken into account. Alternatively, computing a convex region with bounded diameter and the maximum number of points can be done using furthest-site Voronoi diagrams in $O(nk^{2.5}\log k + n\log n)$ time and $O(nk)$ space~\cite{aiks}.
We show that with a dynamic programming method, we can ensure both a maximum area \emph{and} a maximum diameter in $O(n^6k)$ time and $O(n^3k)$ space.
A previous solution to this problem~\cite{umcmodel} takes exponential
time, which made it applicable only to very small point sets. Hence, that paper also presented a heuristic to find good but not necessarily optimal solutions more efficiently.

Experimentally, we investigate three types of point set: (i) uniformly distributed in a fixed square, where we consider different numbers of points (and hence densities); (ii) normally distributed in a fixed square, with different standard deviations; (iii) real-world medical data that was also used in \cite{umcmodel}.

We investigate different maximal diameters and compare our new algorithm to the existing algorithm by Eppstein, Overmars, Rote, and Woeginger~\cite{eorw} that does not use the diameter constraint. In theory, their algorithm has a much better worst-case running time, and it is possible that the found area satisfies the diameter condition without explicitly enforcing it. This would be the case if densest areas by our definition typically yield roundish regions rather than elongated ones. The experiments show that this is not the case, and we do not get the diameter constraint ``for free''. The solutions are elongated; they have a larger diameter than expected and than is allowed in the application.
%\frank{make sure we have a figure of this in the experiments; and refer to it.}

The new algorithm, which enforces a bounded diameter, has a high worst-case running time. Fortunately, the bound on the diameter allows us to perform significant pruning, making the algorithm often faster in practice for not too dense point sets when compared to the algorithm that bounds the area only, despite being a factor $n^3$ slower in theory. Our experiments analyze this for different data sets and different settings of the diameter with respect to the area.

\subparagraph{Related work.}
The identification of dense areas---clusters---in point sets is a common task in many data analysis applications, as witnessed by the immensely popular DBSCAN method~\cite{ester1996density}. It is impossible to review the complete literature; we concentrate on algorithmic results related to this paper. These are 2-dimensional, use a clear definition of a dense area, compute exact results, and admit a running time analysis in the standard algorithmic sense.

Previous research has considered smallest area, perimeter, or diameter convex polygons that contain at least $k$ points. The mentioned paper \cite{eorw} addresses area or perimeter, and can find smallest polygons with $k$ points on the boundary, or $k$ points on the boundary while having no points inside, or $k$ points in total. In all cases, an $O(n^3k)$ time and $O(n^2k)$ space algorithm solves the problem. More generally, the method works whenever the optimization concerns a ``monotone decomposable weight function'': any function where the value for a convex polygon can be obtained from the values of two parts (and possibly the chord used to split in those parts).
Eppstein~\cite{eppstein} later presented an alternative method for the problem of computing a minimum area convex polygon with $k$ points inside that is more efficient when $k$ is sufficiently small. It runs in $O(n^2\log n+n^2k^3)$ time and $O(n\log n+k^3)$ space.

Around the same time, Arkin, Khuller and Mitchell considered geometric
knapsack problems~\cite{akm}, which include computing a minimum area
convex polygon containing $k$ points. They obtain essentially the same
$O(n^3k)$ time dynamic programming algorithm, but include some other
generalizations than the ones in~\cite{eorw}. 
%Their framework does not allow restricting both the area and the diameter.

The problem of computing a convex polygon that contains $k$ points and has minimum diameter uses a different approach. Note that the diameter is not a monotone decomposable weight function, and hence the approach in~\cite{eorw} does not work directly. Aggarwal, Imai, Katoh and Suri~\cite{aiks} propose an $O(k^{2.5}n\log{}k+n\log{}n)$ time and $O(kn)$ space algorithm for this problem using higher-order Voronoi diagrams.

For smallest $k$-point enclosing rectangle or circle algorithms, see~\cite{chan2021smallest,har2005fast}.

Our research was motivated by~\cite{umcmodel}, where the problem of computing a convex region with bounded area and bounded aspect ratio that contains most points was mentioned. The aspect ratio is the diameter-to-width ratio. The paper describes an automated mitosis detection pipeline that includes an algorithm for the problem mentioned. 
The algorithm first cover the region with points into square patches that partially overlap, and then finds solutions in these patches by one of two methods, a brute-force method and a heuristic.
%The one is optimal, that is, it finds the maximum number of points satisfying the area and aspect ratio constraints, but requires exponential time, which turns out to be prohibitive for several of the patches. The other method is a heuristic that incrementally discards convex hull vertices from the convex hull of each patch to decrease its area and improve its aspect ratio, until the remaining convex shape satisfies the constraints.
\medskip

Note that bounding the aspect ratio is not the same as bounding the diameter, but in combination with a bounded area they realize the same goal of finding not too elongated areas. A constant upper bound on the aspect ratio guarantees a constant upper bound on the diameter, but the reverse is not true. If all points are (nearly) collinear, an upper-bounded aspect ratio cannot be attained by a convex hull of any subset (or perhaps only by a singleton point), whereas an upper-bounded diameter solution always exists, and might be desired.

\subparagraph{Organization.} This paper is organized as
follows. In Section~\ref{sec:algorithms} we review the algorithm from
\cite{eorw} that computes a smallest area convex set with $k$ points
inside. We also present our new algorithm that computes a smallest
area convex set with bounded diameter and $k$ points inside, and prove
its correctness and running time.
%The algorithm can also be used for upper-bounding the aspect ratio or lower-bounding the width, because our algorithm checks all antipodal pairs.
In Section~\ref{sec:experiments} we describe the experimental set-up, including the data used and implementation aspects. In Section~\ref{sec:results} we present the results of the experiments and discuss them. These results concern running time, memory used, number of points inside, and realized area and diameter of the solution.

\section{Algorithms}
\label{sec:algorithms}

In this section we present two dynamic programming algorithms to
compute minimum area polygons containing $k$ points. The former
algorithm appeared in \cite{eorw}; we repeat it since our new
algorithm can be seen as an extension of it. The latter is the new
algorithm, which incorporates a given diameter bound. Both algorithms are used in our experiments.

We start with some notation that is used in both algorithms. Let $P$ be a set of $n$ points in the plane, and assume we are searching for a solution containing $2\leq k\leq n$ points. Assume, for ease of description, that no three points in $P$ are collinear. Observe that a minimum area solution with at least $k$ points will always be a minimum area solution with exactly $k$ points, because if the solution had more than $k$ points from $P$, we can always exclude a corner of the convex solution and reduce its area, while reducing the number of points by one only.

For three points $a,b,c$, let $A(a,b,c)$ denote the area of the
triangle $\triangle abc$ whose corners are $a,b,c$, and let $K(a,b,c)$ denote the number of points from $P$ in this triangle. Points of $P$ on the boundary of $\triangle abc$ are not counted.

For a convex polygon $C$, the diameter is the distance between the two points of $C$ that realize the largest distance. Two points that realize the diameter are called a diametral pair. In our context, points will always be points that are in the set $P$, because $C$ is always the convex hull of some subset of the points in $P$. A pair of points $a,b$ of $C$ is antipodal if two parallel lines exist, one through $a$ and one through $b$, such that the interior of polygon $C$ lies between these lines (inside the slab). When the vertices of $C$ are in general position (no two sides are parallel), a convex polygon with $m$ vertices has $m$ antipodal pairs. If there are parallel sides, the number of antipodal pairs can be larger but is still $O(m)$. It is well-known that for a convex polygon $C$, its diameter is realized by an antipodal pair of vertices of $C$.

\begin{figure}[tb]
\centering
\includegraphics{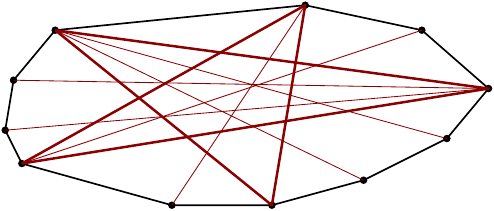}
\caption{The antipodal structure of a convex polygon. Bold chords form a five-pointed star in brown, the six thin brown edges connect to leafs and split star angles.}
\label{fig:antipodal}
\end{figure}

If we assume that no two sides of $C$ are parallel, the antipodal pairs have a clear structure. If we draw them as chords in $C$, then we get one odd-sided star plus zero or more edges that end in leaves, which always split an angle of a star vertex opposite the leaf, see Figure~\ref{fig:antipodal}. This is true because any two chords have an endpoint in common or they intersect properly. Note that an antipodal pair may be the endpoints of an edge that bounds $C$.

The standard method to find all antipodal pairs of a convex polygon in linear time is referred to as rotating calipers~\cite{toussaint1983solving}. It comes down to a traversal of the boundary of $C$ on opposite sides in parallel, using the slopes of the edges to decide on which side we make the next step and find a new antipodal pair. This general idea will be integrated into the dynamic programming algorithm for the bounded-diameter case.

\subsection{Computing a minimum area convex polygon containing $k$ points}

We describe the algorithm of Eppstein, Overmars, Rote and Woeginger~\cite{eorw}.

Take any point $p\in P$ and choose it to be the bottommost vertex of a set of candidate solutions. In other words, we search for the optimal solution among all convex polygons with $k$ points inside that have $p$ as the bottommost vertex. By iterating over all points $p\in P$, we will find the global optimum with a linear factor extra in the time bound.

Given the bottommost point $p$ in the solution, we can use a recursive structure when setting up the dynamic program by considering a bottom-vertex triangulation. This is a triangulation where all vertices on the boundary connect to $p$. A convex polygon $C$ is now partitioned into a fan of triangles around $p$ and above it. 

Assume the vertices are listed in clockwise order around $C$ as $p, p_1, \ldots, p_m$. Then the last (most clockwise) triangle in the fan around $p$ is $\triangle pp_{m-1}p_m$. 

If $C$ is an optimal solution with $k$ points from $P$, then the
convex polygon $p, p_1, \ldots, p_{m-1}$ is an optimal solution with
$k-K(p,p_{m-1},p_m)-1$ points among the solutions that have $p_{m-1}$ as
the counterclockwise neighbor of $p$ and where $p_{m-2},p_{m-1},p_m$
makes a right turn (is convex at $p_{m-1}$ in $C$). This suggests a
3-dimensional table $T[q,r,k]$, defined as: the minimum area of a
convex polygon with $p$ as the bottommost vertex, $q$ as its ccw (counterclockwise)
neighbor, and $r$ as the ccw neighbor of $q$. See
Figure~\ref{fig:eppstein}. This leads to a standard dynamic program with the recurrence:
\[T[q,\,r,\,k] = A(p,q,r) + \min_{s}\;T[r,\,s,\,k-K(p,q,r)-1]\,,\]
where $s$ is chosen over the points above the horizontal line through $p$, left of the directed line from $p$ to $r$, and that leads to convexity at $r$. Since $K(p,q,r)$ counts only points strictly inside $\triangle pqr$, we subtract one more for point $q$. As the base case we let $T[q,\,r,\,k]=A(p,q,r)$ if triangle $\triangle pqr$ has exactly $k$ points, $p$, $q$, and $r$ included.

It is easy to see that filling an entry takes $O(n)$ time, and the whole dynamic program requires $O(n^3k)$ time. Together with the fact that we try all $n$ points as the bottommost point $p$, the whole algorithm takes $O(n^4k)$ time.

In \cite{eorw} a technique is described that reduces the running time by a linear factor. A similar table is used, assuming a given point $p$ as bottommost point, but with a subtly different definition: $T[q,r,k]$ is the smallest area convex polygon with $k$ points inside where $p$ is the bottommost vertex, $q$ is its ccw neighbor, and the convex polygon lies on the same side of the line through $q$ and $r$ as $p$ does. Note that $r$ might not be part of the convex polygon itself.
With a rotational sweep around $q$, entries are filled in only $O(1)$ time per point $r$. We consider all points $r$ above a horizontal line through $p$, not just the ones left of the directed line through $p$ and $q$. We consider the full lines through $r$ and $q$ and order them by angle around $q$, starting clockwise from the line through $p$ and $q$.

\begin{figure}[tb]
  \centering
  \includegraphics{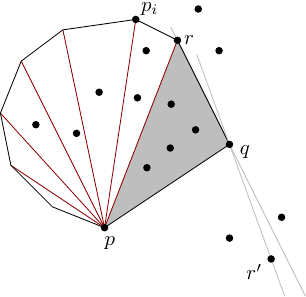}
  \caption{An optimal solution can be decomposed in $\triangle pqr$ and an
    optimal convex polygon with area
    $T[r,\,p_i,\,k-K(p,q,r)-1]$. Table entries can be computed in
    $O(1)$ time by rotating around $q$.  }
  \label{fig:eppstein}
\end{figure}

The rotational sweep maintains a value $T_q^*$ that is the minimum area encountered so far with $k$ points. The sweep has two cases: $(i)$ Point $r$ lies left of the directed line through $p$ and $q$ (in that order), in which case we may use $T[q,\,r,\,k-K(p,q,r)-1]$ to potentially find a smaller $T_q^*$, in which case we update $T_q^*$ and fill $T[q,r,k]$;
$(ii)$ Point $r$ lies to the right of the directed line, in which case we fill
$T[q,r,k]$ with $T_q^*$. This leads to an $O(n^3k)$ time
algorithm. Details can be found in \cite{eorw}.

\subsection{Computing a minimum area convex polygon with bounded diameter containing $k$ points}

In this section we show how to compute a minimum area convex polygon with bounded diameter and $k$ points inside. We do this by trying all pairs of points $p,q,\in P$, assuming they form the diameter, and then computing an optimal convex polygon under this assumption. The algorithm is an extension of the basic, $O(n^4k)$ time algorithm from~\cite{eorw}.
At the end of this section we show how to use this algorithm to find a convex polygon with bounded diameter and area and the maximum number of points inside.
% Omitted proofs are in Appendix~\ref{app:omittedproofs}.

\begin{figure}[htb]
\centering
\includegraphics{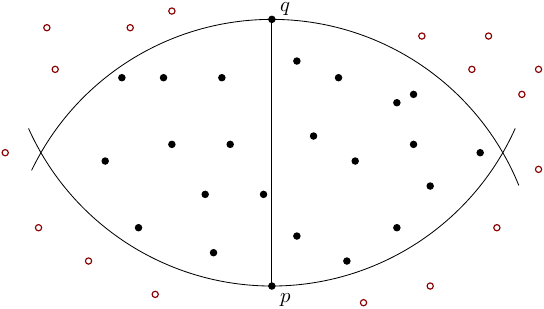}
\caption{The black points lie inside the lune, whereas the red points
  do not. The latter are ignored in the dynamic program that assumes
  $p,q$ is the diameter.}
\label{fig:lune}
\end{figure}

After choosing a pair $p,q$, we first select all points of $P$ that are at most $|pq|$ from both $p$ and from $q$. They are the points in a lune formed by the intersection of two disks with radius $|pq|$, centered at $p$ and at $q$, see Figure~\ref{fig:lune}. For ease of description and without loss of generality, we let $pq$ be vertical with $p$ below $q$. We now need to form the right chain of a convex polygon, clockwise (cw) from $q$ to $p$, and the left chain, clockwise from $p$ to $q$. While doing so, we must make sure that we do not consider solutions that use a point $\ell$ on the left chain that is too far from a point $r$ on the right chain. Note that two points on the right chain can never be too far from each other, and the same is true for the left chain (since we restricted our points to the ones inside the lune).

To solve the problem for a given $p,q$ we combine the dynamic
programming solution from \cite{eorw} with a rotating calipers
algorithm to make sure that all antipodal pairs (one from the left
chain and one from the right chain) have distance no more than $|pq|$.

% [CHANGE] definitions that were on l.213-218 are now here

Let $R[r,r',\ell,k]$ be the minimum area of a convex polygon with $pq$ as a diameter, where $r$ is the ccw neighbor of $p$, $r'$ is the ccw neighbor of $r$, $\ell$ is the ccw neighbor of $q$, that has exactly $k$ points inside, and where $\ell$ is antipodal to $r$ and to $r'$.
Analogously, let $L[\ell,\ell',r,k]$ be the minimum area of a convex polygon with $pq$ as a diameter, where $\ell$ is the ccw neighbor of $q$, $\ell'$ is the ccw neighbor of $\ell$, $r$ is the ccw neighbor of $p$, that has exactly $k$ points inside, and where $r$ is antipodal to $\ell$ and to $\ell'$. 
We allow $r'$ to be $q$ and $\ell'$ to be $p$. Note that $p$ and $q$ are antipodal: there is a pair of horizontal lines through $p$ and through $q$ that have any convex polygon with $p,q$ as diameter in between. We study antipodality more closely in the following lemma.

\begin{restatable}{lemma}{lemAntipodalStructure}
  \label{lem:antipodal_structure}
Let $C$ be a convex polygon with $p,q$ as diameter, with $r$ as ccw neighbor of $p$ and $\ell$ as ccw neighbor of $q$. Furthermore, assume that $C$ has no two sides that are parallel. Then we have at least one of the following cases:
\begin{itemize}
    \item $r$ is antipodal to only $q$, to only $\ell$, or to both but no other points.
    \item $\ell$ is antipodal to only $p$, to only $r$, or to both but no other points.
\end{itemize}
\end{restatable}

\begin{proof}
If $\ell$ is $p$, then the lemma is clear, and the same is true if $r$ is $q$. So we assume that the four vertices $p,q,r,\ell$ are distinct.

Let $r'$ be the ccw neighbor of $r$ and $\ell'$ the ccw neighbor of $\ell$. Using a rotating caliper argument, we know that $p,q$ are antipodal, and either $r$ is antipodal to $q$, or $\ell$ is antipodal to $p$, but not both. This is true because when we rotate the lines through $p$ and $q$ in ccw direction, then either the line through $q$ encounters $\ell$ first, or the line through $p$ encounters $r$ first. If both happen at the same time, then $C$ has parallel sides, violating the assumption.

Since the cases are symmetric, we assume that $r$ is antipodal to $q$. Continuing the rotating caliper argument on the parallel lines through $q$ and through $r$ in ccw direction, either the line through $q$ encounters $\ell$ first, or the line through $r$ encounters $r'$ first. In the latter case, $r$ is antipodal to only $q$. In the former case, we continue the rotating caliper once more: Either $\ell'$ is encountered first, or $r'$ is encountered first. In the former case, $\ell$ is antipodal to only $r$, and in the latter case, $r$ is antipodal to $q$ and $\ell$ but to no other vertices of $C$. 

Together with the symmetric cases, the lemma follows (also when $\ell'=p$ or $r'=q$).
\end{proof}

From the lemma above, it follows that $\ell$ is not antipodal to $r'$ or $r$ is not antipodal to $\ell'$.
Now consider an optimal solution $C_{\opt}$ with a given diameter $pq$, and let $\ell$, $r$, $\ell'$ and $r'$ be as in the proof. If $\ell$ is not antipodal to $r'$, then we can see $C_{\opt}$ as being composed of a triangle $\triangle \ell\ell'q$ and an optimal subsolution $C'_{\opt}$ where $\ell'$ is the ccw neighbor of $q$. Since $\ell$ is antipodal to no more than $r$ and possibly $p$, by induction we can assume that $C'_{\opt}$ has all antipodal pairs no larger than the diameter $|pq|$, we need to test only the distance $|\ell r|$ to conclude that $C_{\opt}$ also has all antipodal pairs close enough. If $r$ is not antipodal to $\ell'$, the argument is symmetric.
%In the former case, we find $C_{\opt}$ as $L(\ell,\ell',r,k)$ and in the latter case we find it as $R(r,r',\ell,k)$.

\begin{figure}[htb]
\centering
\includegraphics{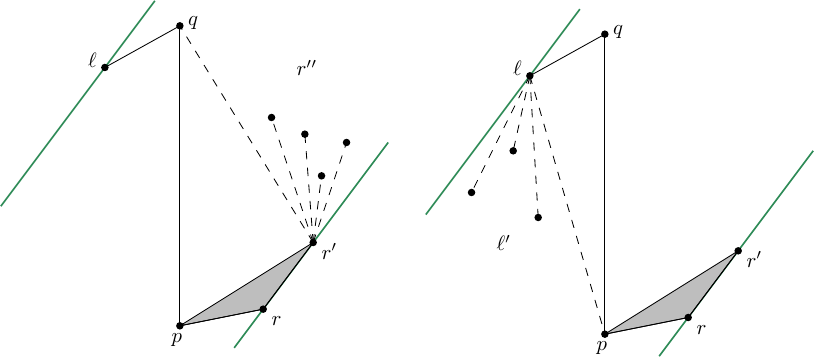}
\caption{The two cases in the recurrence that defines $R[r,r',\ell,k]$. The grey triangle is left out when using the optimal substructure.}
\label{fig:leftright}
\end{figure}

% [CHANGE] definitions on l.213-218 were here

% \frank{Do we allow $\ell$ to be $q$? (i.e. in case all points lie on
%   the right)? That case does not seem entirely well defined. (since we
%   then at some point require that $r''$ is antipodal with $\ell = q$)
% }

% \frank{Let $L'[\ell,\ell',k]$ be the minimum area of a convex polygon
%   with: $pq$ as diameter, $q$ as the ccw neighbor of $p$, $\ell$ is
%   the ccw neighbor of $q$, $\ell'$ the ccw neighbor of $\ell$ that has
%   exactly $k$ points in it.
%   %
%   Note that this subproblem slightly relaxes the condition on $L$;
%   i.e. it does not require $\ell'$ to be antipodal with $r'=q$ anymore.
%   %
%   In the subproblem $R[r,q,\ell,k]$ (i.e. where $r'=q$), we want to
%   switch to the subproblem $A(p,r,q) + \min_{\ell'}
%   L'[\ell,\ell',k-K(p,r,q)-1]$ instead of $L[\ell,\ell',q,\cdot]$.
%   %
% }

If $|r\ell|>|pq|$ then $R[r,.,\ell,.]=\infty$. Similarly, if $|r'\ell|>|pq|$ then $R[.,r',\ell,.]=\infty$.

The recurrence optimizes over two options: use a subproblem of the $R$-type or use a subproblem of the $L$-type. 
%Here we use an optimal R-substructure and an optimal L-substructure, each with their properties. The additional tests and optimizations should ensure that: (1) all options are checked, and (2) all options considered are valid (convex and obeying the diameter constraint). 
The optimization of the $R$-type using a step on the right side is:
\[R[r,\,r',\,\ell,\,k]= A(p,r,r')+ \min_{r''} R[r',\,r''\,,\ell,\,k-K(p,r,r')-1]\,.\]
Here the $r''$ must be such that $r''r'r$ is a right turn at $r'$, vertex $r''$ is right of the line through $r'$ and $q$, and $|r''\ell|\leq |pq|$. In such a case, $r'$ is antipodal to only $\ell$. We can choose $r''$ to be $q$. If $r'=q$, then this option cannot be chosen because the right boundary is finished. Finally, $\ell$ can be $p$ to capture the case that the left side is empty.

The optimization of the $R$-type using a step on the left side is:
\[R[r,\,r',\,\ell,\,k]= A(p,r,r')+ \min_{\ell'} L[\ell,\,\ell'\,,r'\,,k-K(p,r,r')-1]\,.\]
Here, $\ell'$ must be such that $q\ell\ell'$ is a left turn at $\ell$, vertex $\ell'$ is to the left of the line through $\ell$ and $p$ and right of the line through $\ell$ that is parallel to the line through $r$ and $r'$ (showing that $\ell'$ and $r$ are not antipodal), and $|\ell'r'|\leq |pq|$.
We can choose $\ell'$ to be $p$, which would complete the left boundary. Also here, $r'$ might be $q$.

As base case of the recursion, we take all triangles including the
points $p$ and $q$. In particular, we define $R[r_1,q,p,k]=A(p,r_1,q)$
if $k=3+K(p,r_1,q)$, and $\infty$ otherwise.

The definition of $L[\ell,\ell',r,k]$ is symmetric, having analogous
options. Note that since $q$ is the unique highest point and $p$ the
unique lowest point, ``left'' and ``right'' of any line through $p$
(or $q$) and another point is well-defined.

%Our algorithm must allow for the right part or the left part to be empty, in which case the diameter $p,q$ is the right or left boundary of the convex polygon. We can do so by defining $R[r,r',\emptyset,k]$ and $L[\ell,\ell',\emptyset,k]$ as well.

\begin{figure}[htb]
\centering
\includegraphics{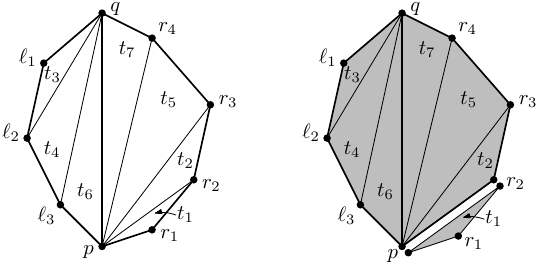}
\caption{A solution $C$ seen as a decomposition into triangles that are consecutively removed.}
\label{fig:dp-order}
\end{figure}

Intuitively, the algorithm constructs an optimal solution as follows,
see Figure~\ref{fig:dp-order}.  Consider a \emph{triangulation} of the
solution $C$ for a given diameter $p,q$ where $p$ has chords to all
vertices on the right and $q$ has chords to all vertices on the
left. Let $p=r_0,r_1,\ldots,r_j=q$ be the vertices incident to edges
on the right boundary. Then the edges $r_1r_2,\ldots,r_{j-1}r_j$ are
the sequences of edges of the triangles in the right half that are
opposite to $p$ in these triangles. Let $q=\ell_0,\ldots,\ell_h=p$ be
defined analogously for the left half.

Consider the sorted sequence (by positive turn angle from the positive
$x$-axis vector) of edges $\vec{e}_1,\ldots,\vec{e}_{j+h-2}$ from
$\overrightarrow{r_1r_2},\ldots,\overrightarrow{r_{j-1}r_j}$ and
$\overrightarrow{\ell_2\ell_1},\overrightarrow{\ell_3\ell_2},\ldots,\overrightarrow{\ell_h\ell_{h-1}}$. This
sorted sequence defines an order on the triangles in the left and
right halves of $C$ that is a merge of the two sequences.  The
solution $C$ can be deconstructed by peeling off the triangles
according to that ordered sequence. If triangle $\triangle pr_1r_2$ is
first in the sequence, then $C$ is composed of this triangle and the
convex polygon $C$ where that triangle is removed. As we argue next,
this solution is found when computing $R[r_1,r_2,\ell_1,k]$, where $k$
is the number of points of $P$ in $C$.

\begin{restatable}{lemma}{lemLAntipodal}
  \label{lem:l1_antipodal}
  Let $C$ be an optimal solution, and let $t_1=\triangle pr_1r_2$ be
  the first triangle in the triangulation as described above. Vertex
  $\ell_1$ is antipodal to $r_1$ and $r_2$.
\end{restatable}

\begin{proof}
  Consider the line $h$ parallel to $\overrightarrow{r_1r_2}$ through
  $\ell_1$, and let $h^+$ be the halfplane not containing $r_1$. We
  will prove that $h^+$ is empty. Since the halfplane right of
  $\overrightarrow{r_1,r_2}$ is also empty (since $r_1r_2$ is an edge
  of $C$), it then follows $\ell_1$ is antipodal to both $r_1$ and
  $r_2$.
  \begin{figure}[bt]
    \centering
    \includegraphics{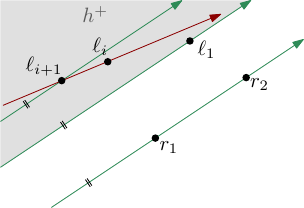}
    \caption{The halfplane $h^+$ must be empty, and thus $\ell_1$ is
      antipodal with $r_1$ and $r_2$.}
    \label{fig:l1_antipodal_proof}
  \end{figure}
  
  Assume, by contradiction, that there is some point, and thus a
  vertex, of $C$ in $h^+$. Let $\ell_{i+1}$ be the vertex in
  $C \cap h^+$ furthest from the line through $\overrightarrow{r_1r_2}$,
  see Figure~\ref{fig:l1_antipodal_proof}. Observe that this is indeed a
  vertex on the left, as either $p$ or $q$ is further from the line through
  $\overrightarrow{r_1r_2}$ than any vertex $r_j$. Consider the line
  parallel to $\overrightarrow{r_1r_2}$ through $\ell_{i+1}$, oriented
  in the same direction as $\overrightarrow{r_1r_2}$. By definition
  of $\ell_{i+1}$, the point $\ell_i$ is to the right of this directed
  line. It now follows that either (i) the positive turn
  angle of $\overrightarrow{\ell_{i+1}\ell_i}$ is smaller than that of
  $\overrightarrow{r_1r_2}$, or (ii) the direction vector of
  $\overrightarrow{r_1r_2}$ is pointing upwards while the direction
  vector of $\overrightarrow{\ell_{i+1}\ell_i}$ is pointing downwards. By
  definition of $t_1$, $\overrightarrow{r_1r_2}$ has minimum turning
  angle, and thus we are in the latter case. However, that means that
  the $y$-coordinate of $\ell_i$ is smaller than the $y$-coordinate of
  $\ell_{i+1}$. Since we make a right turn at every vertex
  $\ell_i,..,\ell_1,q$ while going from $\ell_{i+1}$ towards
  $\ell_0=q$ (and we turn at most $\pi/2$ in total) it follows that
  $\ell_{i+1}$ has a larger $y$-coordinate than $q$. This contradicts
  that $pq$ is vertical and has maximum diameter. We thus conclude
  that $\ell_1$ is the vertex furthest from $\overrightarrow{r_1r_2}$,
  and is thus antipodal with $r_1$ and $r_2$.
\end{proof}

As a consequence of this lemma, $r_1$ is not antipodal to $\ell_2$, so removing $r_1$ from consideration the recursive problem does not remove any antipodal pairs in $C$ that must be checked.

\begin{restatable}{lemma}{lemOptimalConsidered}
    An optimal solution with $pq$ as diameter is considered during the algorithm.
\end{restatable}

\begin{proof}
  Let $C$ be an optimal solution, and let its triangulation be as
  described above. Also assume (by symmetry) that the first triangle
  in the sequence $t_1=\triangle pr_1r_2$. Observe that $C$ thus has:
  (i) $pq$ as a diameter, (ii) $r_1$ as the ccw neighbor of $p$, (iii)
  $r_2$ as the ccw neighbor of $r_1$, (iv) $\ell_1$ as the ccw
  neighbor of $q$, (v) $k$ points inside it (for some
  $k \in \mathbb{N}$), and (vi) $\ell_1$ antipodal with both $r_1$ and
  $r_2$ (by Lemma~\ref{lem:l1_antipodal}). Therefore, by definition of
  $R$, the area of $C$ is $R[r_1,r_2,\ell_1,k]$. Symmetrically, it
  follows that if $t_1$ lies on the left it is a solution to an $L$
  subproblem. We now establish by induction that our definitions of
  $R$ and $L$ are correct.

  If $t_1$ is the only triangle in $C$, it thus follows that we are in
  a base case, and $C=t_1$ indeed has area
  $R[r_1,r_2,\ell_1,k]=A(p,r_1,r_2)=A(p,r_1,q)$ and contains
  $k=3+K(p,r_1,r_2)$ points.

  If there exists a triangle $t_2$, and this triangle also lies on the
  right, then it directly follows from the induction hypothesis that
  the area of $C$ without triangle $t_1$ is
  $R[r_2,r_3,\ell_1,k-K(p,r_1,r_2)-1]$, and thus our definition of
  $R[r_1,r_2,\ell_1,k]$ correctly computes the area of
  $C$. Symmetrically, if $t_2$ lies on the left, the induction
  hypothesis gives us that the area of $C$ without triangle $t_1$ is
  $L[\ell_1,\ell_2,r_2,k-K(p,r_1,r_2)-1]$. It again follows that our
  definition of $R[r_1,r_2,\ell_1,k]$ is correct.
\end{proof}

\begin{restatable}{lemma}{lemNoIlligal}
    No solution is determined that is not convex or where $p,q$ is not a diameter.
  \end{restatable}

\begin{proof}
  By construction. Since all antipodal pairs are tested and are no
  longer than the diameter, no pair of points can have a larger
  distance. Convexity is by construction.
\end{proof}

The two tables $R$ and $L$ require $O(n^3k)$ storage, each entry is filled in $O(n)$ time if we know the point count in each triangle, and there are $O(n^2)$ choices for $p,q$. Hence we obtain:

\begin{restatable}{theorem}{thmResult}
    Given a set of $n$ points in the plane, a diameter bound $D$, and an integer $k\geq 3$, a smallest convex polygon that has diameter at most $D$ and contains $k$ points inside or on the boundary can be computed in $O(n^6k)$ time and $O(n^3k)$ space. 
\end{restatable}
\begin{proof}
    The algorithm that attains this bound has already been described, except for determining how many points of $P$ lie in any triangle, as used in the recursion. We can do an easy preprocessing step where we determine for each of the $O(n^2)$ edges how many points of $P$ lie vertically below that edge. This allows us to retrieve the point count inside a triangle in constant time by using the three edges of the triangle, and adding or subtracting these counts. This takes $O(n^3)$ time, which is clearly dominated by the time to run the dynamic programming steps.
\end{proof}

The final algorithm from \cite{eorw} can fill in table entries in
constant time. Unfortunately, we were unable to use a similar approach
(e.g.\ rotating around $r$ filling in the $R$-entries for every point
$r'$) as the table entry depends on a second point $\ell'$ on the left
to guarantee the bound on the diameter. It seems hard to maintain the
appropriate point $\ell'$ as $r'$ varies. Also, rotating around $r$ and $\ell$ simultaneously did not seem to work.

By running the algorithm with some value $k$, we also determine the smallest solutions for all values $3,\ldots,k$. This allows us to assume an maximum diameter $D$ and maximum area $A$, and determine the convex polygon with most points inside within the same time bound.
Start with $k=4$, run the algorithm above, and if the area is smaller than $A$, double $k$ and repeat. If the area is larger than $A$ for the current $k$, we can look up the largest point count for which the area is at most $A$.

\begin{corollary}
Given a set $P$ of $n$ points in the plane, a diameter bound $D$ and an area bound $A$, a convex polygon with diameter at most $D$ and area at most $A$ that contains the maximum number of points of $P$ can be computed in $O(n^6k)$ time and $O(n^3k)$ space, where $k$ is the number of points in the optimal solution.
\end{corollary}

\section{Experimental setup}
\label{sec:experiments}

The second part of this paper contains the experimental part of the
research. We answer a few questions relating to the resources used by
the algorithm from~\cite{eorw}, referred to as the area-only (A)
algorithm, and the new version that also enforces an upper bound on
the diameter, referred to as the area-diameter (AD)
algorithm. We use AD$_{d}$ to refer to the AD algorithm with maximum diameter set to $d$.
When testing on data from the application itself, we also compare with the algorithm from~\cite{umcmodel}, referred to as the area selection (AS) algorithm like in their paper. We are interested in answering the following questions:\smallskip
\begin{itemize}
    \item How efficient in terms of running time and storage requirements are the A  and AD$_{d}$ algorithms? How does the AD$_d$ algorithm's efficiency depend on the diameter?
    \item What diameters are realized by the A algorithm in the optimal solutions?
    \item How do the number of enclosed points compare in the optimal solutions of the A and AD$_d$ algorithms?
    \item What are the found results in terms of area, diameter, and point count, of the three algorithms, now including the AS algorithm~\cite{umcmodel}, on real-world data of the application that motivated this research?  
\end{itemize}

Note that the first question relates to the resources used by the algorithms, whereas the other three questions relate to the output features.
To answer these questions, we implemented the algorithms and selected data to be used.
In the next subsections we discuss the data and the implementation aspects. The next section discusses the results and answers the research questions.

\subsection{Data and settings}

We use three types of data for the experiments:
\begin{itemize}
    \item Uniformly distributed point sets inside a square of size $20\times 20$ mm$^2$ of increasing cardinalities $100,\,110, \ldots, 200$.
    \item Normally distributed point sets of 100 points centered inside a square of size $20\times 20$ mm$^2$ of increasing standard deviations ranging from $\sigma=0.5, 1.0, 1.5, \ldots, 6.5$ mm. Points outside the square are discarded and re-generated.
    \item Medical data used in~\cite{umcmodel}. We selected ten data sets with up to $700$ points, in particular the ten most populated sets.
\end{itemize}

We will always search for an area of at most $4$ mm$^2$.\footnote{While \cite{umcmodel} lists 2 mm$^2$ as the desired area of the region, personal communication with the authors of the paper showed that they use 4 mm$^2$ currently in the application.} For the AD algorithm, we examine maximum diameters of $2$, $3$, $4$, $5$, and $6$ mm. Note that for a diameter of $2$ mm, the maximum area possible is $\pi$ mm$^2$, so in this setting, we cannot ``use'' the full allowed area of $4$ mm$^2$.

Since the first two types of data are randomly generated, we generated $100$ point sets in each setting and report averages over these $100$ runs.

The real-world data is obtained by preprocessing the data set provided by the Department of Digital Pathology at UMC Utrecht, containing $1982$ point sets representing the mitotic figures identified in Whole-Slide-Images (WSIs) from actual patients~\cite{umcmodel}. The identification was carried out by the classification model introduced in \cite{aimodel}, which produced $2$-dimensional points with detection accuracies, since the points have a different likeliness of being a mitotic cell. We apply a preprocessing step to filter out any point with detection accuracy below $0.86$ and convert the coordinates from pixels to millimeters. We consider the top $10$ most populated point sets for our experiments. Table \ref{tab:real-data-overview} provides an overview of this data. The average nearest neighbor distances provide an indication of how clustered the points are.

Both the synthetic generation and preprocessing of real-world data are fully automated in a Python script which we provide in a repository.

\begin{table}
  \caption{The ten real-world point sets used in our experiments. We report the data set index (I), number of points $n$, and Average Nearest Neighbor distance (ANN) in mm with standard deviation.}
    \centering
    \begin{tabular}{ccccccc}
    \toprule
    I & $n$ & ANN & \hspace*{1cm}  &  I & $n$ & ANN  \\
    \midrule
    0 & 699 & $0.21\pm0.14$ && 5 & 474 & $0.30\pm0.23$ \\
    1 & 623 & $0.32\pm0.22$ && 6 & 412 & $0.35\pm0.23$  \\
    2 & 594 & $0.28\pm0.19$ && 7 & 382 & $0.34\pm0.23$ \\
    3 & 526 & $0.29\pm0.21$ && 8 & 377 & $0.33\pm0.27$ \\
    4 & 492 & $0.32\pm0.25$ && 9 & 362 & $0.48\pm0.34$  \\
    \bottomrule
    \end{tabular}
\label{tab:real-data-overview}
\end{table}

\subsection{Implementation}

We implemented the A algorithm from~\cite{eorw} and the AD algorithm
in C++\footnote{Our code is available at \url{https://github.com/gianmarcopicarella/area-diameter-algorithm}.}. We implemented
all other necessary data structures and algorithms that are not
already in the std ourselves. The implementations do not use
parallelism and run in a single thread. The experiments were run on a
MacBook Pro using an Apple M1 Pro (3228 MHz) and 16GB (6400 MT/s
LPDDR5) RAM. It runs macOS Sonoma 14.2.1, and we compiled it for 64
bits using clang 19.1.6 with the -O3 option. We used the Google
Benchmark library to set up our experiments and perform data
collection. The data is then post-processed in Python with several
scripts that we provide in our repository.  Regarding the Area
Selection method (AS), we used the Python implementation made available by
\cite{umcmodel}. We use this implementation mostly to compare output,
not to compare resources, although we report times.

The AS algorithm from \cite{umcmodel} first partitions the data in patches of $3\times 3$ mm, allowing overlap so that solutions across the boundary of a patch are not missed. For each patch, if it contains at most $25$ points, a brute-force method determines the optimum. Otherwise, a heuristic is used that takes the convex hull of a patch and incrementally discards points from the convex hull, until the criteria are met. More precisely, it considers each convex hull point $p$, analyzes the decrease in area and cardinality if it were discarded, but also if it were replaced by any point inside the triangle formed by $p$ and its two neighbors (details in \cite{umcmodel}).

The A and AD algorithms use a few optimizations.
First, we preprocess for quickly determining the point count in triangles, so that the dynamic program can avoid spending more than constant time. We use only $O(n^2)$ time for preprocessing and storage by computing the number of points below each of the $\binom{n}{2}$ possible edges.
To represent all table entries of the AD algorithm, we used hashing, since many table entries are not used. In particular, we use an array with $k$ hash tables for $R$ and for $L$: to retrieve $R[r,r',\ell,k']$ we inspect the $k'$-th hash table and use the indices of $r$, $r'$ and $\ell$ as the key, and similar for $L$.
In the A algorithm, most table entries are used, so the A algorithm does not benefit as much from hashing. Hence it is not used.
Finally, both the A and AD algorithms use an early exit strategy. For the AD algorithm point pairs are treated in decreasing order of number of points in the lune. If a solution of cardinality $k$ is known, then all point pairs whose lune contains fewer points can be skipped.
For the A algorithm, the early exit strategy is much less strong, but we can still skip processing points $p$ as bottommost point if there are not enough points above it to improve the current best solution.
Any preprocessing time is included in the timings. The implementations handle degenerate cases like collinear points properly.

\section{Experimental results}
\label{sec:results}

We first present and discuss the results on synthetic data, starting
with the resources needed by the A and AD algorithms. We then discuss the output results. Finally, we discuss the results on real-world data.

% ----

\begin{figure}
\centering
    \begin{subfigure}{0.49\textwidth}
        \centering
        \includegraphics[width=0.93\textwidth]{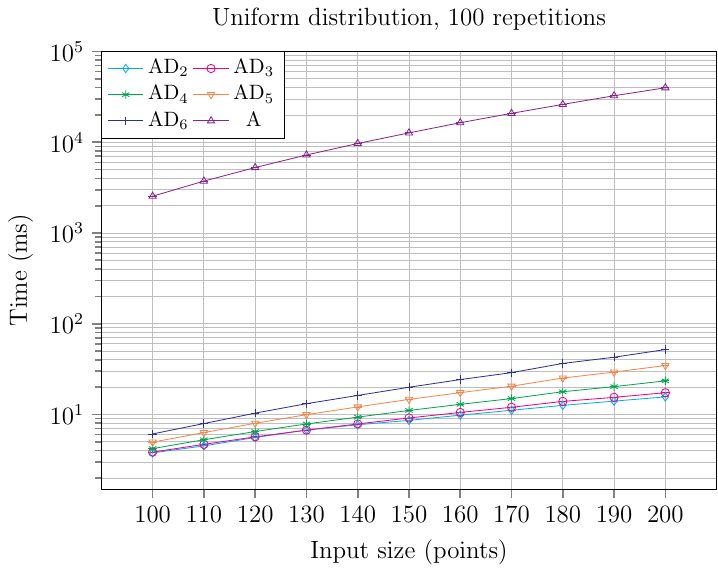}
    \end{subfigure}
    \hfill
    \begin{subfigure}{0.49\textwidth}
        \centering
        \includegraphics[width=0.93\textwidth]{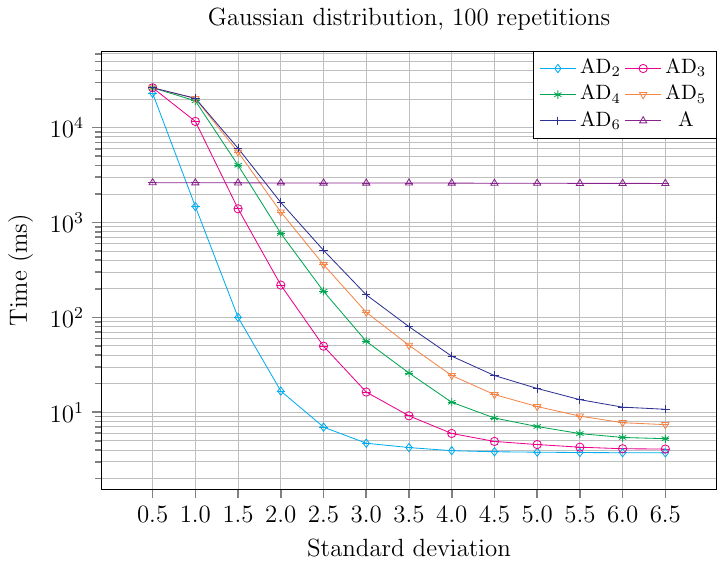}
    \end{subfigure}

    \caption{Left: average running time over $100$ uniformly distributed data sets of sizes $100$ to $200$. Right: average running time over $100$ normally distributed data sets of $100$ points with varying standard deviations.}
    \label{fig:uni-gauss-time-small}
\end{figure}

\begin{figure}
\centering
    \begin{subfigure}{0.49\textwidth}
        \centering
        \includegraphics[width=0.93\textwidth]{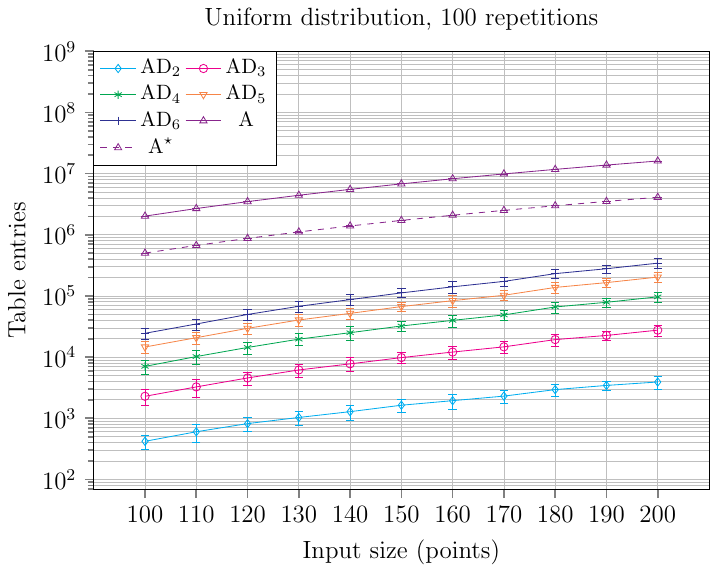}
    \end{subfigure}
    \hfill
    \begin{subfigure}{0.49\textwidth}
        \centering
        \includegraphics[width=0.93\textwidth]{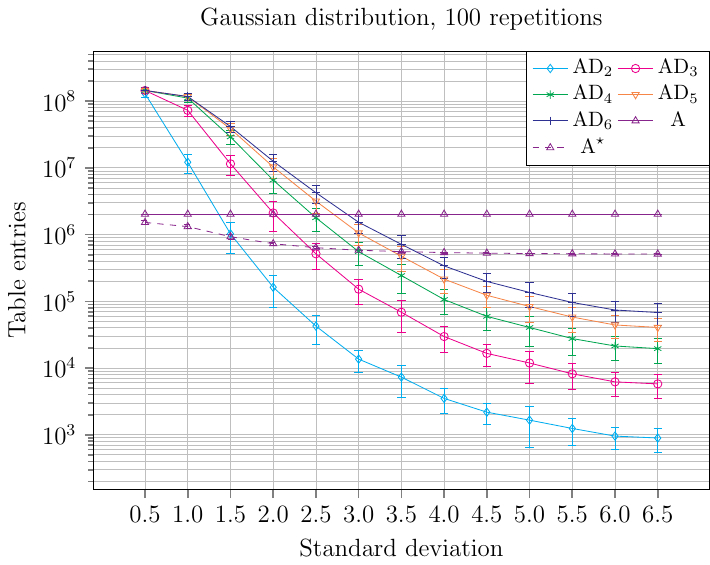}
    \end{subfigure}

    \caption{Left: average number of entries over $100$ uniformly distributed data sets of sizes $100$ to $200$. Right: average number of entries over $100$ normally distributed data sets of $100$ points with varying standard deviations.}
    \label{fig:uni-gauss-entries-small}
\end{figure}

\subparagraph{Resources.}
%In Figure~\ref{fig:uni-gauss-time-small}, the running time for uniformly and normally distributed point sets is displayed.
In Figure~\ref{fig:uni-gauss-time-small}, the running time for synthetic point sets generated with uniform and normal (Gaussian) distributions is displayed. We see the AD algorithm makes very good use of the bounded diameter when compared with the A algorithm. Despite requiring $O(n^6k)$ time worst case, the pruning possible by bounding the diameter in the AD algorithm easily outweighs the much better $O(n^3k)$ running time of the A algorithm. This is true for the point set sizes tested. The situation is different for the normal distribution. When $\sigma \leq 1.5$, the $100$ points are clustered around the center and the AD algorithm must try many more options, leading to a worse performance than the A algorithm. Both for the uniform and the Gaussian distribution, we can clearly see that the AD algorithm is more efficient for smaller diameters than for larger diameters, which is expected.
In Appendix~\ref{app:moredense} we show results on larger point sets, up to $500$ points, for the uniform distribution, but due to computational demands we averaged over only $5$ repetitions. The patterns appear the same; the AD algorithms are still more efficient.

To measure the memory usage of the algorithms we consider the number of table entries that are used by the algorithms. See
Figure~\ref{fig:uni-gauss-entries-small}. The results are similar to those of the running times. In the basic implementation of the A algorithm, a fixed size table is allocated, and the number of entries does not depend on the specific point set, only on its size. To also let the A algorithm benefit from fewer entries, we consider a version of the A algorithm that stores repeated entries found during the rotational sweep around $q$ just once. Whenever we use the same $T^*_q$ to fill $T[q,r,k]$, we can store it as a sequence of values that are all $T^*_q$, at the expense of extra query time for table entries. This is shown by the dashed purple curve. We would use a factor $3$--$4$ fewer entries, except for the Gaussian distribution with lower values of $\sigma$.  We also studied the memory usage instead of the number of entries in the tables, but since
this showed the same patterns, we left out these graphs.

\subparagraph{Output.}
We can compare the results of running the algorithms on cardinality (point count), area, and diameter. Since all algorithms are optimal for their settings, it is clear that the cardinality for the A algorithm is highest, followed by the cardinalities of the AD algorithm by decreasing diameter.

%The output numbers cardinality, area, and diameter, averaged, are shown in Figures~\ref{fig:uni-count-small}--\ref{fig:uni-diameter-small} for the uniform distribution and in
%Figures~\ref{fig:gau-count}--\ref{fig:gau-diameter} for the Gaussian distribution.

The output numbers cardinality, area, and diameter, averaged, are shown in Figures~\ref{fig:uni-gauss-count-small}--\ref{fig:uni-gauss-diameter-small} for the uniform and normal distributions.

Indeed the A algorithm finds larger cardinality solutions; the differences between different diameter bounds are also clear. It is perhaps surprising that the highest cardinality solutions are typically not compact, but rather elongated. This is confirmed in Figure~\ref{fig:uni-gauss-diameter-small}, where the A algorithm always finds solutions with diameters $8$--$18$ mm. With an area close to $4$ mm$^2$, these regions must be elongated. Appendix~\ref{app:moredense} shows the results for larger point sets, which confirm the findings up to $500$ points, albeit based on just $5$ repetitions.

\begin{figure}
\centering
    \begin{subfigure}{0.49\textwidth}
        \centering
        \includegraphics[width=0.9\textwidth]{Images/Experiments/Synthetic/uniform_20x20_100-200points_100rep/entries_uniform.pdf}
    \end{subfigure}
    \hfill
    \begin{subfigure}{0.49\textwidth}
        \centering
        \includegraphics[width=0.9\textwidth]{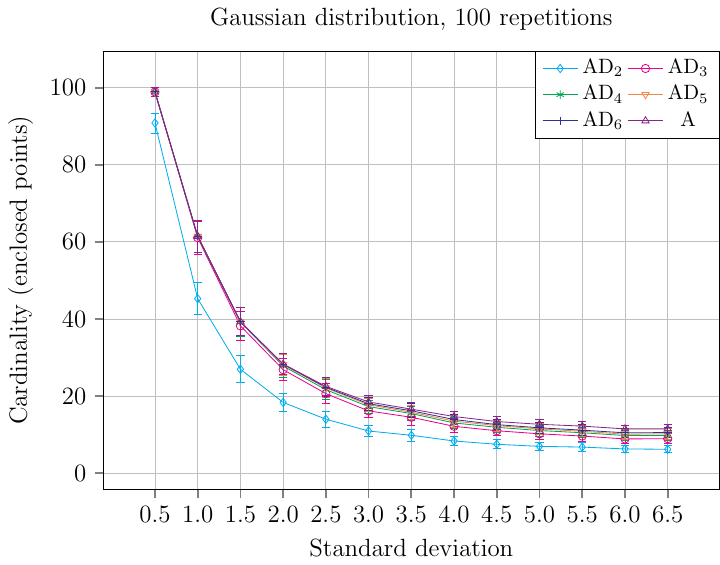}
    \end{subfigure}

    \caption{Left: average cardinality over $100$ uniformly distributed data sets of sizes $100$ to $200$. Right: average cardinality over $100$ normally distributed data sets of $100$ points with varying standard deviations.}
    \label{fig:uni-gauss-count-small}
\end{figure}

\begin{figure}
\centering
    \begin{subfigure}{0.49\textwidth}
        \centering
        \includegraphics[width=0.93\textwidth]{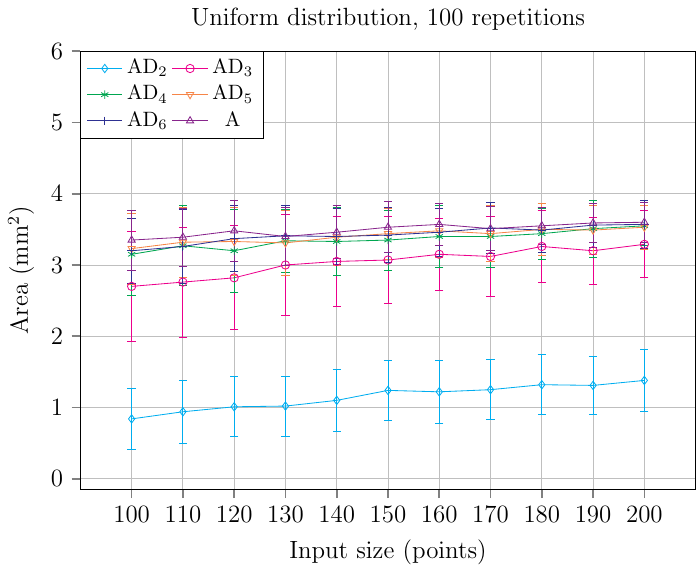}
    \end{subfigure}
    \hfill
    \begin{subfigure}{0.49\textwidth}
        \centering
        \includegraphics[width=0.93\textwidth]{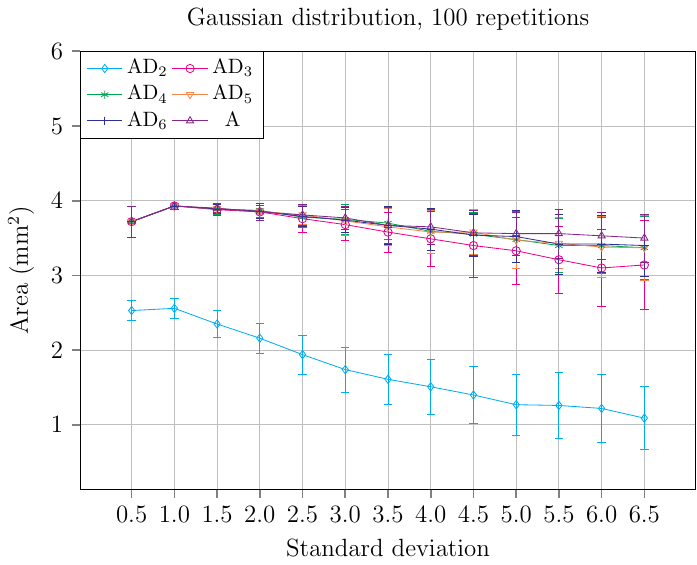}
    \end{subfigure}

    \caption{Left: average area over $100$ uniformly distributed data sets of sizes $100$ to $200$. Right: average area over $100$ normally distributed data sets of $100$ points with varying standard deviations.}
    \label{fig:uni-gauss-area-small}
\end{figure}

\begin{figure}
\centering
    \begin{subfigure}{0.49\textwidth}
        \centering
        \includegraphics[width=0.93\textwidth]{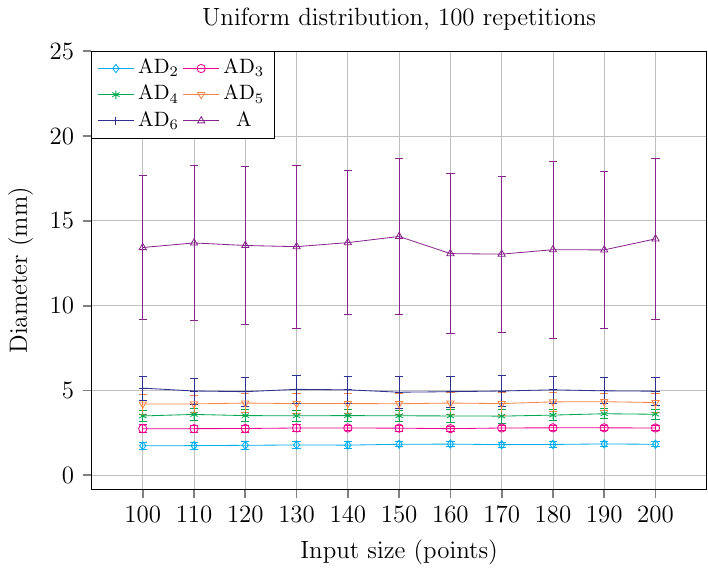}
    \end{subfigure}
    \hfill
    \begin{subfigure}{0.49\textwidth}
        \centering
        \includegraphics[width=0.93\textwidth]{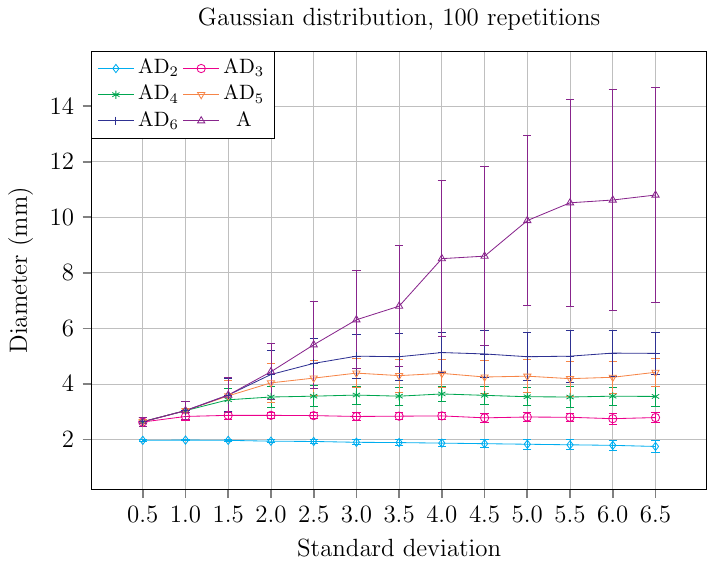}
    \end{subfigure}

    \caption{Left: average diameter over $100$ uniformly distributed data sets of sizes $100$ to $200$. Right: average diameter over $100$ normally distributed data sets of $100$ points with varying standard deviations.}
    \label{fig:uni-gauss-diameter-small}
\end{figure}

\begin{table}
  \caption{The runtime (seconds) of the Area-Diameter (AD) algorithm, the area selection (AS) algorithm and the A algorithm applied to the real-world data summarized in Table~\ref{tab:real-data-overview}. AS is run for step sizes $s=1.98$ and $s=0.5$; A is run within each patch generated with step size $s=0.5$.}
    \centering
    \begin{tabular}{ccccc}
    \toprule
    I & \textbf{AD$_4$} & \textbf{AS ($s=1.98$)} & \textbf{AS ($s=0.5$)} & \textbf{A ($s=0.5$)}\\
    \midrule
    0 & 494.19 & \textbf{5.24}  & 72.83  & 13.7 \\
    1 & 14.7   & 17.15 & 274.82 & \textbf{2.84} \\
    2 & 29.38  & 19.89 & 218.36 & \textbf{3.66} \\
    3 & 24.43  & 19.12 & 263.07 & \textbf{3.17} \\
    4 & 8.83   & 13.83 & 229.33 & \textbf{2.08} \\
    5 & 21.76  & 9.16  & 141.13 & \textbf{4.51} \\
    6 & 2.08   & 9.32  & 161.07 & \textbf{1.63} \\
    7 & 6.19   & 13.66 & 174.71 & \textbf{1.82} \\
    8 & 37.18  & 9.66  & 77.41  & \textbf{2.58} \\
    9 & \textbf{0.19}   & 4.83  & 52.02  & 1.47 \\
    \bottomrule
    \end{tabular}
\label{tab:real-data-time-comparison}
\end{table}

\subparagraph{Real-world data.}
We consider the output of the algorithms on the data that inspired
this research, where we include the area selection (AS) algorithm from~\cite{umcmodel}.
That algorithm uses patches for efficiency reasons, and to be able to run an exact, exponential-time algorithm on patches that contain few points. Interestingly, patches are also a way to limit the diameter that can be found by the A algorithm, namely at most the diameter of a patch. Since patches are $3\times 3$ mm, the A algorithm and the AS algorithm may find regions with diameter up to $4.24$ mm. We use a diameter of $4$ mm in the AD algorithm, which does not use patches.

Patches are shifted with overlap to cover the whole region. We show results of the AS algorithm with a rather large step size of $1.98$ mm, and results of the A and AS algorithms with a step size of $0.5$ mm.\footnote{The step size of $1.98$ mm is based on the actual usage of the algorithm currently by the authors of \cite{umcmodel}.}

The results are shown in Tables~\ref{tab:real-data-time-comparison} and~\ref{tab:comparison}. In the ten examples, the four algorithms usually find the same rough region, but seldom the exact same convex solution. One example is shown in Figure~\ref{fig:4alg}; all ten data sets are in Appendix~\ref{app:realdata}. In two cases, the AS algorithm with $s=1.98$ finds a very different region, confirming that that step size of a moving window matters and must be set small enough.
The results of the two AS methods show how many more points can be enclosed by using a smaller step size. The results of AS and A with step size $s=0.5$ show how well the AS heuristic works compared to the optimal algorithm A (optimal per patch). The timings show that A is faster than AS, while being optimal per patch as well. While AD$_4$ finds better solutions without patches, it is clearly slower.

\begin{table}[t]
    \caption{Comparison of the results obtained by applying the
      area-diameter algorithm (AD), the area selection method (AS) and
      the A algorithm to all patches of the point
      sets summarized in Table~\ref{tab:real-data-overview}. AS is run for
      step sizes $s=1.98$ and $s=0.5$; A is run for step size $s=0.5$. All
      the values for area and diameter are expressed in
      $\text{mm}^2$ and $\text{mm}$ respectively. Best counts are shown bold; values $k$ and $\hat{k}$ use different diameter bounds and are not directly comparable.
    }
    \centering
    %\resizebox{\columnwidth}{!}{
    \begin{tabular}{c@{\hskip 3\tabcolsep}ccc@{\hskip 3\tabcolsep}ccc@{\hskip 3\tabcolsep}ccc@{\hskip 3\tabcolsep}ccc}
        \toprule
        I & \multicolumn{3}{@{}c@{\hskip 3\tabcolsep}}{\textbf{AD$_4$}} &
            \multicolumn{3}{@{}c@{\hskip 3\tabcolsep}}{\textbf{AS ($s=1.98$)}} &
            \multicolumn{3}{@{}c@{\hskip 3\tabcolsep}}{\textbf{AS ($s=0.5$)}} &
            \multicolumn{3}{@{}c@{\hskip 3\tabcolsep}}{\textbf{A ($s=0.5$)}} \\
      \cmidrule(lr{\dimexpr 3\tabcolsep+0.5em}){2-4}
      \cmidrule(lr{\dimexpr 3\tabcolsep+0.5em}){5-7}
      \cmidrule(lr{\dimexpr 3\tabcolsep+0.5em}){8-10}
      % \cmidrule(lr{\dimexpr 1\tabcolsep+0.5em}){11-13}
      \cmidrule(lr){11-13}
        & $k$ & area & dia. & $\hat{k}$ & area & dia. & $\hat{k}$ & area & dia. & $\hat{k}$ & area & dia. \\
        \midrule
        0 & \textbf{76} & 3.99 & 3.15 & 68 & 3.42 & 2.89 &70 & 3.74 & 2.88 & 74 & 3.9 & 3.07\\
        1 &\textbf{42} & 3.83 & 3.99 &  36 & 3.84 & 2.74 & 39 & 3.97 & 2.9 & 40 & 3.94 & 2.87\\
        2 & \textbf{39} & 3.93 & 2.72 &  33 & 3.99 & 2.85 & \textbf{39} & 3.98 & 2.66 & \textbf{39} & 3.93 & 2.72\\
        3 & \textbf{42} & 3.94 & 3.1 &  34 & 3.99 & 2.71 & 39 & 3.97 & 3.15 & 40 & 3.88 & 3.1\\
        4 &\textbf{42} & 3.88 & 2.99 &  38 & 3.87 & 2.82 & 40 & 3.99 & 2.83 & \textbf{42} & 3.98 & 2.88\\
        5 & \textbf{44} & 3.88 & 3.7 & 37 & 3.82 & 3.31 & 39 & 3.76 & 3.17 & 39 & 3.82 & 3.01\\
        6 &\textbf{31} & 3.98 & 3.85 & 27 & 3.72 & 2.73 & 28 & 4.0 & 3.45 & 28 & 3.94 & 3.38\\
        7 &\textbf{34} & 3.89 & 3.22 &  30 & 3.94 & 2.65 & 32 & 3.92 & 3.22 & \textbf{34} & 3.92 & 3.22\\
        8 & \textbf{61} & 3.91 & 3.94 & 57 & 3.86 & 2.67 &59 & 3.97 & 2.72 & 60 & 3.92 & 3.01\\
        9 &\textbf{28} & 3.86 & 3.49 & 23 & 3.97 & 3.41 &27 & 3.97 & 2.74 & 27 & 3.97 & 2.74\\
        \bottomrule
    \end{tabular}
    %}
    \label{tab:comparison}
\end{table}

\begin{figure}
\centering
    \begin{subfigure}{0.24\textwidth}
        \centering
        \includegraphics[width=\linewidth]{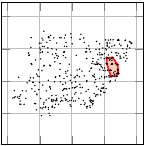}
    \end{subfigure}
    \hfill
    \begin{subfigure}{0.24\textwidth}
        \centering
        \includegraphics[width=\linewidth]{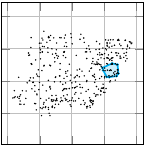}
    \end{subfigure}
    \hfill
    \begin{subfigure}{0.24\textwidth}
        \centering
        \includegraphics[width=\linewidth]{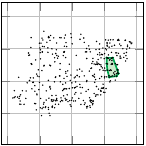}
    \end{subfigure}
    \hfill
    \begin{subfigure}{0.24\textwidth}
        \centering
        \includegraphics[width=\linewidth]{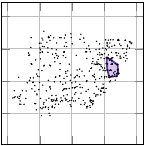}
    \end{subfigure}
\caption{Convex areas found by AD (a), AS using $s=1.98$ (b), AS using $s=0.5$ (c) and A using $s=0.5$ (d) for point set $7$. The line spacing is just for reference and is set to $5$ $\text{mm}$.}
\label{fig:4alg}
\end{figure}

\section{Conclusions and further research}

We have shown that maximum cardinality convex subsets with bounded area and diameter can be computed in polynomial time by a dynamic programming method that embeds rotating calipers. The algorithm runs in $O(n^6k)$ time, a cubic factor higher than the algorithm that does not consider the diameter.

% [CHANGE] "These are interesting theoretical questions." -> "These are open theoretical questions for future research.". I had to replace "rather similar" with "similar" to prevent the creation of a new line.

The original application worked with a bounded area and aspect ratio, which is similar but not quite the same as area and diameter. It is unclear whether a polynomial-time algorithm exists that computes maximum cardinality subsets with bounded area and aspect ratio, or bounded area and perimeter. These are open theoretical questions for future research.

We have also shown that the new algorithm may well be faster than the algorithm that ignores diameter in practical situations, because optimal solutions without a diameter constraint are often very elongated. With a bounded diameter, we can often prune significantly. Of course this depends heavily on the specific data set. We have analyzed this experimentally for uniformly distributed data and Gaussian distributed data.

For the data from the application, we have shown that we can compute an optimal solution in the practical situation within a reasonable amount of time, where the problem is phrased as a well-defined optimization problem and the solution need not use patches. 
%With better hardware and more optimized algorithms, it should be possible to reduce the running times significantly and process larger point sets.

\bibliography{references}

\clearpage
\appendix

\section{Results for uniformly distributed larger point sets}
\label{app:moredense}

We report on additional experiments that involve larger point sets, up to $500$ points using a uniform distribution. Due to computation time, we have only $5$ repetitions at the moment, so the averages are less consistent.

We can see that the AD algorithms are still more efficient in time than the A algorithm. For the number of table entries, we can see that the results get considerably closer with increasing numbers of points.

The results on the output features cardinality, area, and diameter are similar although they are more erratic, due to the small number of repetitions.
In particular, we notice that the optimal solutions found by the A algorithm, without a limit on the diameter, are still much more elongated.

\begin{figure}[h!p]
    \centering
    \vspace*{1cm}
    \includegraphics[scale=0.90]{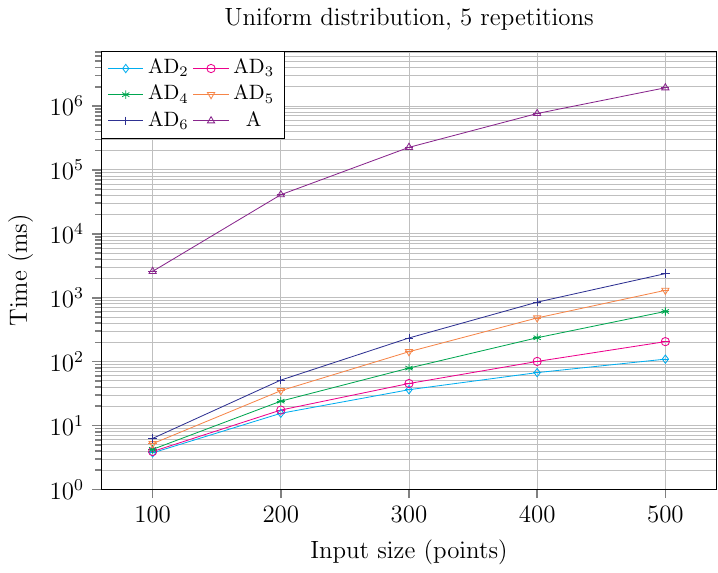}
    \caption{Average running time over $5$ uniformly distributed data sets of sizes $100$ to $500$.}
    \label{fig:uni-time}
\end{figure}

\clearpage

%\begin{figure}
%    \centering
%    \includegraphics[]{Images/Experiments/Synthetic/memory_uniform.pdf}
%    \caption{}
%    \label{fig:uni-mem}
%\end{figure}

\begin{figure}
    \centering
    \includegraphics[scale=0.9]{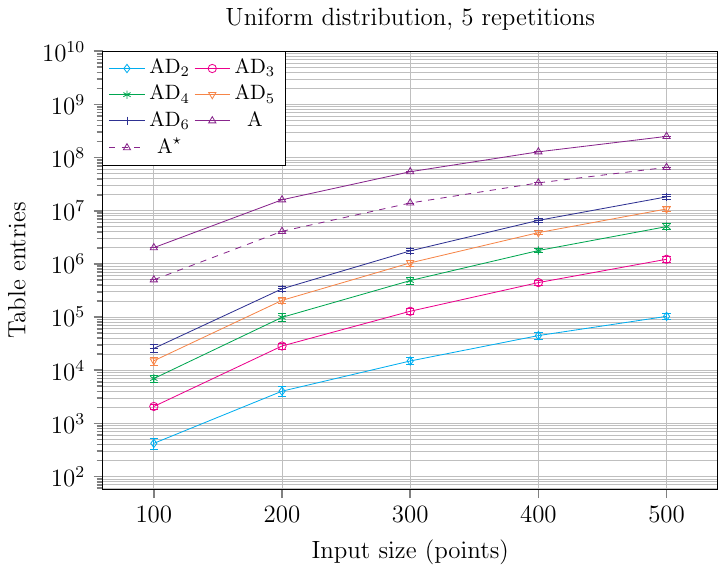}
    \caption{Average number of entries over $5$ uniformly distributed data sets of sizes $100$ to $500$.}
    \label{fig:uni-entries}
\end{figure}

\begin{figure}
    \centering
    \includegraphics[scale=0.9]{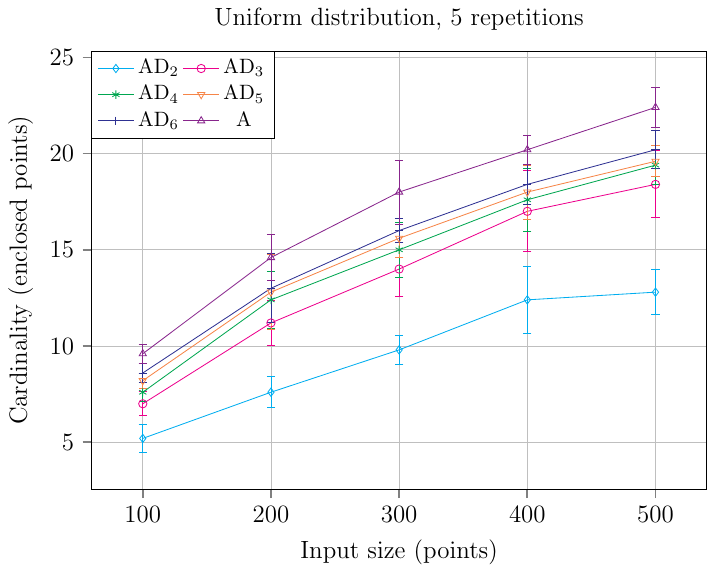}
    \caption{Average cardinality over $5$ uniformly distributed data sets of sizes $100$ to $500$.}
    \label{fig:uni-count}
\end{figure}

\begin{figure}
    \centering
    \includegraphics[scale=0.9]{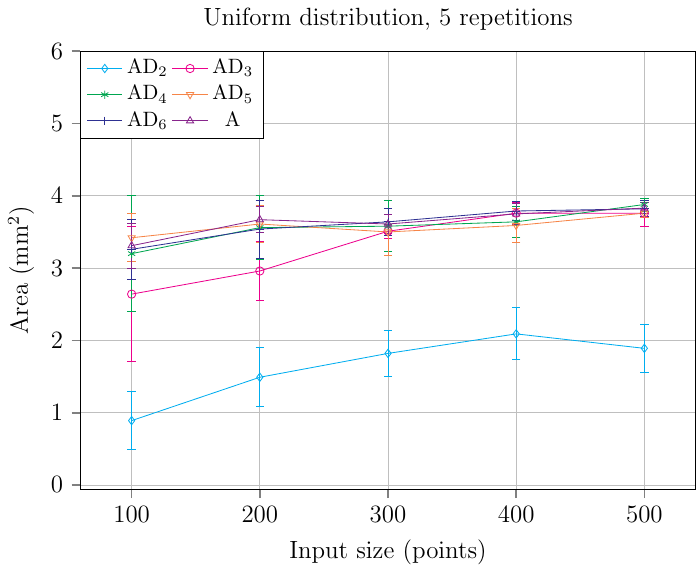}
    \caption{Average area over $5$ uniformly distributed data sets of sizes $100$ to $500$.}
    \label{fig:uni-area}
\end{figure}

\begin{figure}
    \centering
    \includegraphics[scale=0.9]{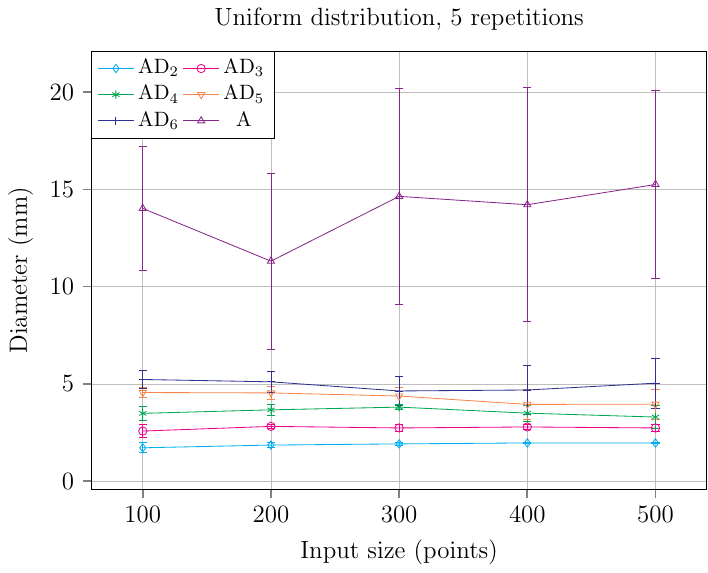}
    \caption{Average diameter over $5$ uniformly distributed data sets of sizes $100$ to $500$.}
    \label{fig:uni-diameter}
\end{figure}

\clearpage

\section{Visual output for the real-world data sets}
\label{app:realdata}

\begin{figure}[H]
\centering
\small
    \begin{subfigure}{0.24\textwidth}
        \centering
        \includegraphics[width=\linewidth]{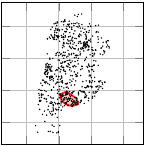}
    \end{subfigure}
    \hfill
    \begin{subfigure}{0.24\textwidth}
        \centering
        \includegraphics[width=\linewidth]{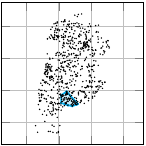}
    \end{subfigure}
    \hfill
    \begin{subfigure}{0.24\textwidth}
        \centering
        \includegraphics[width=\linewidth]{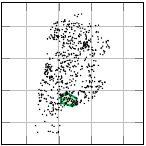}
    \end{subfigure}
    \hfill
    \begin{subfigure}{0.24\textwidth}
        \centering
        \includegraphics[width=\linewidth]{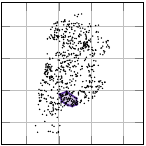}
    \end{subfigure}

    \bigskip

    \begin{subfigure}{0.24\textwidth}
        \centering
        \includegraphics[width=\linewidth]{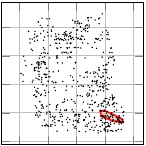}
    \end{subfigure}
    \hfill
    \begin{subfigure}{0.24\textwidth}
        \centering
        \includegraphics[width=\linewidth]{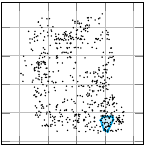}
    \end{subfigure}
    \hfill
    \begin{subfigure}{0.24\textwidth}
        \centering
        \includegraphics[width=\linewidth]{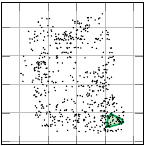}
    \end{subfigure}
    \hfill
    \begin{subfigure}{0.24\textwidth}
        \centering
        \includegraphics[width=\linewidth]{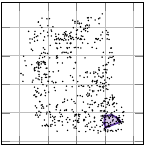}
    \end{subfigure}

    \bigskip

    \begin{subfigure}{0.24\textwidth}
        \centering
        \includegraphics[width=\linewidth]{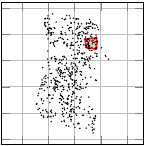}
    \end{subfigure}
    \hfill
    \begin{subfigure}{0.24\textwidth}
        \centering
        \includegraphics[width=\linewidth]{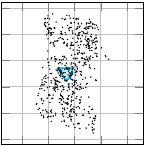}
    \end{subfigure}
    \hfill
    \begin{subfigure}{0.24\textwidth}
        \centering
        \includegraphics[width=\linewidth]{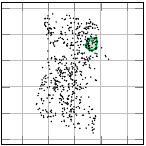}
    \end{subfigure}
    \hfill
    \begin{subfigure}{0.24\textwidth}
        \centering
        \includegraphics[width=\linewidth]{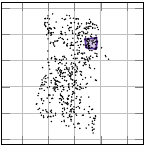}
    \end{subfigure}

    \bigskip

    \begin{subfigure}{0.24\textwidth}
        \centering
        \includegraphics[width=\linewidth]{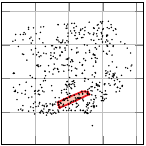}
        %\caption{}
    \end{subfigure}
    \hfill
    \begin{subfigure}{0.24\textwidth}
        \centering
        \includegraphics[width=\linewidth]{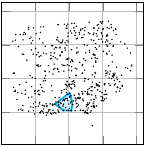}
        %\caption{}
    \end{subfigure}
    \hfill
    \begin{subfigure}{0.24\textwidth}
        \centering
        \includegraphics[width=\linewidth]{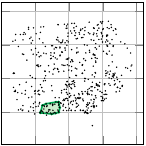}
        %\caption{}
    \end{subfigure}
    \hfill
    \begin{subfigure}{0.24\textwidth}
        \centering
        \includegraphics[width=\linewidth]{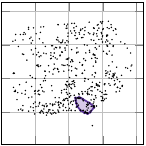}
        %\caption{}
    \end{subfigure}

    \bigskip

    \begin{subfigure}{0.24\textwidth}
        \centering
        \includegraphics[width=\linewidth]{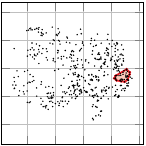}
    \end{subfigure}
    \hfill
    \begin{subfigure}{0.24\textwidth}
        \centering
        \includegraphics[width=\linewidth]{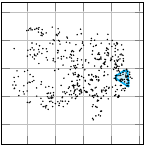}
    \end{subfigure}
    \hfill
    \begin{subfigure}{0.24\textwidth}
        \centering
        \includegraphics[width=\linewidth]{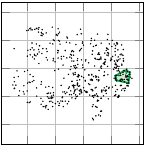}
    \end{subfigure}
    \hfill
    \begin{subfigure}{0.24\textwidth}
        \centering
        \includegraphics[width=\linewidth]{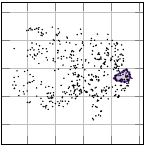}
    \end{subfigure}

    \caption{Convex areas found by $\text{AD}_4$ (red, left column), AS using $s=1.98$ (blue), AS using $s=0.5$ (green) and A using $s=0.5$ (purple) for the real-world point sets with index $I=0,1,2,3,4$. The line spacing is set to $5$ $\text{mm}$.}
    \label{plot:real-solutions-1}
\end{figure}

\begin{figure}[H]
\centering
\small
    \begin{subfigure}{0.24\textwidth}
        \centering
        \includegraphics[width=\linewidth]{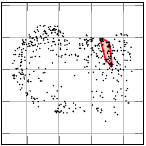}
    \end{subfigure}
    \hfill
    \begin{subfigure}{0.24\textwidth}
        \centering
        \includegraphics[width=\linewidth]{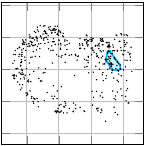}
    \end{subfigure}
    \hfill
    \begin{subfigure}{0.24\textwidth}
        \centering
        \includegraphics[width=\linewidth]{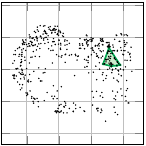}
    \end{subfigure}
    \hfill
    \begin{subfigure}{0.24\textwidth}
        \centering
        \includegraphics[width=\linewidth]{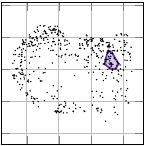}
    \end{subfigure}

    \bigskip

    \begin{subfigure}{0.24\textwidth}
        \centering
        \includegraphics[width=\linewidth]{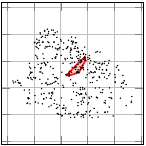}
    \end{subfigure}
    \hfill
    \begin{subfigure}{0.24\textwidth}
        \centering
        \includegraphics[width=\linewidth]{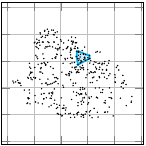}
    \end{subfigure}
    \hfill
    \begin{subfigure}{0.24\textwidth}
        \centering
        \includegraphics[width=\linewidth]{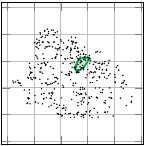}
    \end{subfigure}
    \hfill
    \begin{subfigure}{0.24\textwidth}
        \centering
        \includegraphics[width=\linewidth]{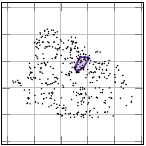}
    \end{subfigure}

    \bigskip

    \begin{subfigure}{0.24\textwidth}
        \centering
        \includegraphics[width=\linewidth]{Images/Experiments/Real/ad_7.pdf}
        %\caption{}
    \end{subfigure}
    \hfill
    \begin{subfigure}{0.24\textwidth}
        \centering
        \includegraphics[width=\linewidth]{Images/Experiments/Real/as0_7.pdf}
        %\caption{}
    \end{subfigure}
    \hfill
    \begin{subfigure}{0.24\textwidth}
        \centering
        \includegraphics[width=\linewidth]{Images/Experiments/Real/as1_7.pdf}
        %\caption{}
    \end{subfigure}
    \hfill
    \begin{subfigure}{0.24\textwidth}
        \centering
        \includegraphics[width=\linewidth]{Images/Experiments/Real/ep_7.pdf}
        %\caption{}
    \end{subfigure}

    \bigskip

    \begin{subfigure}{0.24\textwidth}
        \centering
        \includegraphics[width=\linewidth]{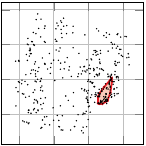}
    \end{subfigure}
    \hfill
    \begin{subfigure}{0.24\textwidth}
        \centering
        \includegraphics[width=\linewidth]{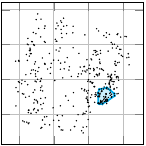}
    \end{subfigure}
    \hfill
    \begin{subfigure}{0.24\textwidth}
        \centering
        \includegraphics[width=\linewidth]{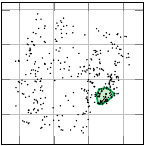}
    \end{subfigure}
    \hfill
    \begin{subfigure}{0.24\textwidth}
        \centering
        \includegraphics[width=\linewidth]{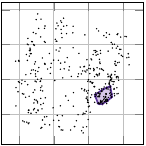}
    \end{subfigure}

    \bigskip

    \begin{subfigure}{0.24\textwidth}
        \centering
        \includegraphics[width=\linewidth]{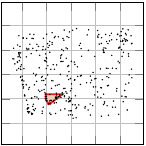}
        %\caption{}
    \end{subfigure}
    \hfill
    \begin{subfigure}{0.24\textwidth}
        \centering
        \includegraphics[width=\linewidth]{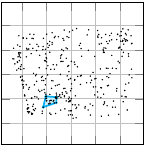}
        %\caption{}
    \end{subfigure}
    \hfill
    \begin{subfigure}{0.24\textwidth}
        \centering
        \includegraphics[width=\linewidth]{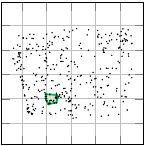}
        %\caption{}
    \end{subfigure}
    \hfill
    \begin{subfigure}{0.24\textwidth}
        \centering
        \includegraphics[width=\linewidth]{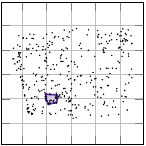}
        %\caption{}
    \end{subfigure}
    
    \caption{Convex areas found by $\text{AD}_4$ (red), AS using $s=1.98$ (blue), AS using $s=0.5$ (green) and A using $s=0.5$ (purple) for the real-world point sets with index $i=5,6,7,8,9$. The line spacing is set to $5$ $\text{mm}$.}
    \label{plot:real-solutions-2}
\end{figure}

\end{document}